\documentclass[letterpaper, 10 pt, conference]{ieeeconf} 
\IEEEoverridecommandlockouts
\usepackage{amsmath,amssymb}
\usepackage{graphicx}
\usepackage{bm}
\usepackage{svg}
\usepackage{textcomp}
\usepackage{xcolor}
\usepackage{cite}
\usepackage{url}
\usepackage[hidelinks]{hyperref}
\usepackage{algorithm}
\usepackage[noend]{algorithmic}
\usepackage{xprintlen}
\usepackage{cases}
\usepackage{tabularx,ragged2e}
\usepackage{booktabs}
\usepackage{nccmath}
\usepackage{caption}
\usepackage{float}
\usepackage[labelformat=simple]{subcaption}
\usepackage{mathtools}
\usepackage{verbatim}
\usepackage{import}

\def\BibTeX{{\rm B\kern-.05em{\sc i\kern-.025em b}\kern-.08em
    T\kern-.1667em\lower.7ex\hbox{E}\kern-.125emX}}
\usepackage{array}
\newcolumntype{P}[1]{>{\centering\arraybackslash}p{#1}}

\newtheorem{problem}{\textbf{Problem}}

\newtheorem{definition}{\bf{Definition}}
\newtheorem{lemma}{\bf{Lemma}}
\newtheorem{theorem}{\bf{Theorem}}
\newtheorem{remark}{\bf{Remark}}
\newtheorem{assumption}{\bf{Assumption}}

\newtheorem{corollary}{\bf{Corollary}}
\newtheorem{example}{\textbf{Example}}

\def\myunderbar#1{\underline{\sbox\tw@{$#1$}\dp\tw@\z@\box\tw@}}

\newcommand{\sign}{\text{sign}}
\newcommand{\vect}[1]{\boldsymbol{#1}}
\newcommand{\norm}[1]{ \lVert #1 \rVert}
\newcommand{\card}[3]{| \mathcal{#1}_{#3}^{#2} |}

\newcommand{\Nhop}[2]{\mathcal{N}_{#1}^{{#2}\text{-hop}}} 
\newcommand{\Neigh}[2]{\mathcal{N}_{#1}^{{#2}}} 
\newcommand{\estimate}[3]{\hat{{#1}}^{{#2}}_{{#3}}}
\newcommand{\state}[3]{{{#1}}^{{#2}}_{{#3}}}
\newcommand{\error}[3]{\Tilde{{#1}}^{{#2}}_{{#3}}}
\newcommand{\statediff}[5]{{#1}^{#2}_{#3}-{#4}^{#5}_{#3}}

\newcommand{\kc}[2]{{#1}_{#2}^{\text{kc}}}


\begin{document}

\title{Communication-aware Multi-agent Systems Control Based on $k$-hop Distributed Observers}
\author{Tommaso Zaccherini, Siyuan Liu and Dimos V. Dimarogonas
	\thanks{This work was supported in part by the Wallenberg AI, Autonomous Systems and
Software Program (WASP) funded by the Knut and Alice Wallenberg (KAW) Foundation, the Horizon Europe EIC project SymAware (101070802), 
the ERC LEAFHOUND Project, and the Swedish Research Council (VR), and Digital Futures.
}
\thanks{Tommaso Zaccherini, Siyuan Liu, and Dimos V. Dimarogonas are with the Division
of Decision and Control Systems, KTH Royal Institute of Technology, Stockholm, Sweden.
	E-mail: {\tt\small \{tommasoz,siyliu,dimos\}@kth.se}. 
}
}

\maketitle
\begin{abstract}
We propose a distributed control strategy to allow the control of a multi-agent system requiring $k$-hop interactions based on the design of distributed state and input observers. In particular, we design for each agent a finite time convergent state and input observer that exploits only the communication with the $1$-hop neighbors to reconstruct the information regarding those agents at a $2$-hop distance or more. We then demonstrate that if the $k$-hop based control strategy is set-Input to State Stable with respect to the set describing the goal, then the observer information can be adopted to achieve the team objective with stability guarantees.
\end{abstract}


\section{Introduction}

A multi-agent system refers to a system composed of multiple interacting autonomous agents with their own goals, capabilities, and decision-making processes that work together or compete to achieve collective or individual objectives. Due to their advantages with respect to single agents in term of redundancy and flexibility, they have been extensively investigated during the years under several aspects \cite{eb184279-05f5-3acc-a300-750c6f4a17e8, 1333204, SEYBOTH2013245,NI2010209}.
In particular, thanks to their possibility of performing multiple simultaneous actions, they represents a valid choice to better accomplish the assigned objective in terms of timing and efficiency.

The main drawback compared to single-agent systems consists in the increased complexity in terms of coordination and communication requirements. 
Due to the lack of centralized global memory, the cooperation among the agents relies only on the local information available by means of the inter-agent communication and sensing with the $1$-hop neighbors, while usually the the goal depends on the global state of the system. Therefore, when communication and sensing capabilities are limited, it may become helpful to enable each agent to exploit the estimates of the states of those agents that lie outside its immediate $1$-hop neighborhood. 

Several works concerning distributed state estimation in network systems are available in the literature \cite{ACIKMESE20141037,7440802,7386579}. In \cite{ACIKMESE20141037}, a decentralized observer for a system of agents with discrete-time dynamic is proposed, where knowledge of the model and local information are exploited to estimate the plant state by means of a consensus based filter. In \cite{7440802} instead, an asymptotic observer in which each agent estimates the global state by exploiting the communication interaction with its $1$-hop neighbors is proposed. The main drawback of these approaches is the need to estimate the entire system state, irrespective of the specific information required by each agent. As a result, in large-scale network applications where agents may only require partial information about the global state, these approaches become difficult to implement due to their poor scalability with respect to the number of agents. Alternative approaches rely on the decomposition of the network into subsystems with local controllers where the information exchange between subsystems may or not be allowed as in \cite{VACCARINI2009328} and \cite{yang2006networked}, respectively.
Other results regarding decentralized observers include \cite{kar2010gossip}, where distributed Kalman filters are adopted, \cite{GRIP20121339} where both linear and non-linear interconnected systems are considered and \cite{8754717}, which introduces a distributed, finite-time convergent $k$-hop observer, where each agent in the network only needs to communicate with its $1$-hop neighbors to estimate its own state and input, as well as those of the agents which resides up to $k$-hops away. 

This paper is inspired by the $k$-hop observer-based distributed control strategy proposed in \cite{8754717}. 
By the observation that each agent in the network knows its own state and input and may be able to receive those of its $1$-hop neighbors through communication, we develop a distributed observer by restricting the estimation only to the states and inputs of those agents which are $2$-hop distant or more.
Compared to the work in \cite{8754717}, our results increase the speed of convergence of each agent observer estimation by exploiting the value, provided by the common $1$-hop neighbors, of the states of those elements that are $2$-hop distant. In particular, we propose a finite time $k$-hop distributed observer for non-linear systems where each agent estimates the states and the inputs of the agents $2$-hop distant or more, by interacting only with the agents belonging to its $1$-hop neighborhood. Moreover, we show that under a bounded input assumption the state observer convergence results to be independent of the input observer dynamics. We also show that by adopting a set-ISS feedback control law it is possible to exploit the states estimation information to drive the system towards an equilibrium representing the team objective.
Furthermore, given the reduced number of estimates that each observer needs to update compared to those in \cite{8754717}, the proposed solution results more scalable while dealing with large-scale networks.

The paper is organized as follow: Section \ref{Preliminaries} presents the preliminaries and the problem setting. Section \ref{state_observer} and Section \ref{input_observer} respectively propose the results concerning the $k$-hop distributed state and input observers. In Section \ref{Closed_loop} we introduce the feedback control structure and provide the conditions that guarantee the convergence of a $k$-hop estimation based feedback controller toward the team objective. In Section \ref{Simulations} we provide a simulation result to demonstrate the convergence towards the goal when $k$-hop estimation is used in the feedback controller, and in Section \ref{Conclusion_and_Future_work} we provide final remarks and future work.

\par

\section{ Preliminaries and Problem setting}\label{Preliminaries}
\textbf{Notation}: We denote by $\mathbb{R}$ and $\mathbb R_{\ge 0}$ the set of real and non-negative real numbers, respectively. Let $|S|$ be the cardinality of a set $S$, $\mathbb{R}^n$ be an $n$-dimensional Euclidean space and $\mathbb{R}^{n \times m}$ be a space of real matrices with $n$ rows and $m$ columns. Denote by $I_n$ the identity matrix of size $n$ and by $1_n = [1, \dots,1]^\top$ the vector of ones of size $n$. Given a matrix $B \in \mathbb{R}^{n\times n}$, we represent with $\lambda_i(B)$, $\lambda_{\text{min}}(B)$ and $\lambda_{\text{max}}(B)$ respectively the $i$-th, minimum and maximum eigenvalues belonging to the spectrum $\sigma(B)$ of matrix $B$. Given a positive definite matrix $B=B^\top \succ 0$ and a vector $x \in \mathbb{R}^n$, $\norm{x}_B= \sqrt{x^\top B x}$, with the convention $\norm{x}=\norm{x}_I$. Additionally, $\norm{x}_1 = \sum_{i =1}^n |x_i|$. Given a matrix $B$, we adopt $B \prec 0$ to denote that $B$ is negative definite.
Let $\text{diag}(a_1, \dots, a_n)$ be the diagonal matrix with diagonal elements  $a_1, . . . , a_n$ and let $\otimes$ be the Kronecker product. 
We denote by $\text{sign}(x)$ the non-smooth function defined as: $\sign(x) = 1 \text{ if } x \geq 0$ and $\sign(x) = -1 \text{ if } x < 0$.
Given the presence of the $\sign(\cdot)$ function in the observers' dynamics, non-smooth analysis is required to study their convergence. For this purpose, as in \cite[(1.2a)]{4048815}, we denote by $K[f]: \mathbb{R}^m \rightarrow \mathbb{R}^m$ the set-valued map of a measurable, locally bounded function $f(y): \mathbb{R}^m \rightarrow  \mathbb{R}^m$, the function  defined as $K[f](y) := \bigcap_{\delta>0}\bigcap_{\mu\{M\}=0}\overline{\text{co}}\{ f(\mathcal{B}(y,\delta)/M)\}$,
where $\mathcal{B}(y,\delta)$ denotes the ball of radius $\delta$ centered at $y$, $\cap_{\mu\{M\}=0}$ the intersection over all sets $M$ of Lebesgue measure zero, $\mathcal{B}(y,\delta)/M$ the set difference between $\mathcal{B}(y,\delta)$ and $M$ and $\overline{\text{co}}$ the convex closure.
Moreover, we further define $\norm{K[f]}= \sup_{z \in R(K[f])} \norm{z}$,
where $R(K[f]) =  \bigcup_{y \in \mathbb{R}^m} K[f](y)$.
We use notations $\mathcal{K}$ and $\mathcal{KL}$ to denote the different classes of comparison functions, as follows: $\mathcal{K} \!=\! \{\gamma : \mathbb R_{\ge 0}\rightarrow\mathbb R_{\ge 0}   :  \gamma \text{ is continuous, strictly increasing and } \gamma(0)=0\}$;  $\mathcal{KL} \!=\! \{\beta : \mathbb R_{\ge 0} \!\times \mathbb R_{\ge 0} \rightarrow\mathbb R_{\ge 0}  :$ for each fixed $s$, the map  $\beta(r,s)$  belongs to class  $\mathcal{K}$  with  respect to  $r$  and, for each fixed  nonzero $r$,  the map $\beta(r,s)$ is decreasing with respect to  $s$  and $\beta(r,s) \rightarrow 0 \text{ as } s \rightarrow \infty \}$.

\subsection{Multi-agent systems}
Consider a multi-agent system composed of a set of $n$ interacting agents $\mathcal{V} = \{1,2,\dots, n\}$. Suppose each agent $i \in \mathcal{V}$ behaves according to the nonlinear dynamics:
\begin{equation}
    \label{eq: agent's dynamic}
    \dot x_i(t) =  f(x_i) + A x_i + u_i,
\end{equation}
where $x_i \in \mathbb{X}\subset\mathbb{R}^N$ with $\mathbb{X}$ denoting the state space, $A \in \mathbb{R}^{N \times N}$, $f:\mathbb{R}^N \rightarrow \mathbb{R}^N $ is a Lipschitz nonlinear function with Lipschitz constant $l_{f}$ and $u_i \in \mathbb{U} \subset \mathbb{R}^N$ is a measurable and essentially locally bounded function satisfying Assumption~\ref{Assumption on input}. 
\begin{assumption} One of the following conditions holds:
\label{Assumption: input assumption}
    \begin{enumerate}
        \item $u_i$ is bounded with known upper bound $d_{u_i} \in \mathbb R_{\ge 0} \ \forall i \in \mathcal{V}$.\label{Assumption on bounded input}
        \item The derivative $K[\dot{u}_i](\cdot)$ is bounded with known upper bound $d_{\dot{u}_i} \in \mathbb R_{\ge 0} \ \forall i \in \mathcal{V}$. \label{Assumption on bounded generalized input derivative}
    \end{enumerate} 
    \label{Assumption on input}
\end{assumption}
Furthermore, denote with $\vect{x}$ and $\vect{u}$ the stacked vector of agents states and inputs:
\begin{equation}
\label{eq:vector of state}
    \vect{x}= \left[x_1^\top, \dots, x_n^\top \right]^\top, \ \vect{u}=\left[u_1^\top, \dots, u_n^\top\right]^\top.
\end{equation}
\subsection{Communication Graph}
Let the communication capabilities among the agents be described by an undirected graph $\mathcal{G}=(\mathcal{V},\mathcal{E})$, where $\mathcal{V}$ is the set of nodes and $\mathcal{E} \subseteq \mathcal{V} \times \mathcal{V}$ is the set of edges representing the set of established communication among the agents.

Define a path between agent $i \in \mathcal{V}$ and $j  \in \mathcal{V}$ as the set of non-repeating edges through which $j$ can be reached by $i$. Under this definition, a $k$-hop path between agents $i, j  \in \mathcal{V}$ is a path involving $k$ edges from $i$ to  $j$. 

Denote with $\Nhop{i}{k}$ the set of $k$-hop neighbors of agent $i \in \mathcal{V}$, i.e., the set of nodes $j \in \mathcal{V}$ for which there exists a $p$-hop path from $j$ to $i$ with $2\leq p\leq k$. Moreover, denote the elements of this set with $\Nhop{i}{k} = \{ n^i_1, \dots, n^i_{\eta_i}\}$, where each $n^i_j$ with $j \in \{1, \dots, \eta_i\}$ is the global index of the $j$-th $k$-hop neighbor of $i$ and where $\eta_i= \card{N}{k\text{-hop}}{i}$ represents its cardinality.
For brevity, we indicate with $\Neigh{i}{}$ the set of $1$-hop neighbors of agent $i\in \mathcal{V}$.

For the purpose of the observer design, denote with $\vect{x}^i$ and $\state{\vect{u}}{i}{}$ the vectors containing respectively the state $x_j$ and input $u_j$ of the agents $j \in \Nhop{i}{k}$:
\begin{equation}
    \label{stack vector of real values estimated by agent i}
    \state{\vect{x}}{i}{} = \left[\state{x}{\top}{n_1^i}, \dots ,\state{x}{\top}{n^i_{\eta_i}} \right]^\top, \ \state{\vect{u}}{i}{} = \left[\state{u}{\top}{n_1^i}, \dots ,\state{u}{\top}{n^i_{\eta_i}} \right]^\top
\end{equation}
and let:
\begin{equation}
\label{definition of input and state estimations done by agent i}
    \begin{split}
        \estimate{\vect{x}}{i}{} &= \left[(\estimate{x}{i}{n_1^i})^\top, \dots ,(\estimate{x}{i}{n^i_{\eta_i}})^\top \right]^\top, \\
         \estimate{\vect{u}}{i}{} &= \left[(\estimate{u}{i}{n_1^i})^\top, \dots ,(\estimate{u}{i}{n^i_{\eta_i}})^\top \right]^\top
    \end{split}
\end{equation}
be the vectors containing their estimate carried out by the agent $i$, i.e. $\estimate{x}{i}{n_p^i}$ and $\estimate{u}{i}{n_p^i}$ for $p \in \{1, \dots, \eta_i\}$ are the estimates of the state $\state{x}{}{n_p^i}$ and input $\state{u}{}{n_p^i}$ of agent $n_p^i \in \Nhop{i}{k}$ done by $i$. Furthermore, denote with  $\tilde{\vect{x}}^i$ and $\tilde{\vect{u}}^i$ the estimation errors:
\begin{equation}
    \label{Definition: error on input and state estimation}
    \error{\vect{x}}{i}{} = \statediff{\vect{x}}{i}{}{\hat{\vect{x}}}{i},\ \ \ \ \error{\vect{u}}{i}{} =\statediff{\vect{u}}{i}{}{\hat{\vect{u}}}{i}.
\end{equation}

Indicate with $\vect{x}_i$ and $\vect{u}_{i}$ the vectors defined as: 
\begin{equation}
    \vect{x}_i = 1_{\eta_i} \otimes x_i, \ \ \ \vect{u}_i = 1_{\eta_i} \otimes u_i,
\end{equation}
and denote with $\estimate{\vect{x}}{}{i}$ and $\estimate{\vect{u}}{}{i}$ the stacked vector estimates of 
$x_i$ and $u_i$ computed by each of the agents $j \in \Nhop{i}{k}$:
\begin{equation}
    \begin{split}
    \estimate{\vect{x}}{}{i} &= \left[(\estimate{x}{n_1^i}{i})^\top, \dots , (\estimate{x}{n_{\eta_i}^i}{i})^\top \right]^\top \\
    \estimate{\vect{u}}{}{i} &= \left[(\estimate{u}{n_1^i}{i})^\top, \dots , (\estimate{u}{n_{\eta_i}^i}
    {i})^\top \right]^\top.
    \end{split}
\end{equation}
Similar to \eqref{Definition: error on input and state estimation}, we can define the estimation errors on $\estimate{\vect{x}}{}{i}$ and $\estimate{\vect{u}}{}{i}$, computed by all  $j \in \Nhop{i}{k}$ as:

\begin{equation}
    \label{Definition: error on input and state estimation of agent i}
   \error{\vect{x}}{}{i} = \statediff{\vect{x}}{}{i}{\hat{\vect{x}}}{}, \ \ \ \error{\vect{u}}{}{i} =\statediff{\vect{u}}{}{i}{\hat{\vect{u}}}{}. 
\end{equation}

For future implementation, define for each $i \in \mathcal{V}$ the matrix $\hat{\vect{P}}_i$ as the one selecting the states estimated by agent $i$, $\state{\vect{x}}{i}{} = \hat{\vect{P}}_i \vect{x}$,
where:
\begin{equation}
    \label{eq: definition of the matrix P}
    \hat{\vect{P}}_i = \hat{P}_i \otimes I_{N}
\end{equation}
and $\hat{P}_i = [e_{n_1^i} \ e_{n_2^i} \dots  e_{n^i_{\eta_i}}]^\top$ is a $\eta_i \times n$ binary matrix where each $e_{j}$ with $j \in \Nhop{i}{k}$ is a vector with all zeros except from the $j$-th element which is equal to 1.
In a similar way, let $\vect{P}_i$ be the $\card{N}{}{i} \times n$  matrix selecting the components of the states of the $1$-hop neighbors of agent $i$. 

Example \ref{Example: communication graph = path graph} is reported to clarify the adopted notation.
\begin{example}
    \label{Example: communication graph = path graph}
    Consider a network of 4 agents communicating according to the path graph in Fig.~\ref{fig: 2 fig 3-hop neihbours communication graph = path graph}. Suppose $N = 1$ and that we are interested in estimating 3-hop neighbor agents' state.
    \begin{figure}[t]
        \vspace{0.1cm}
        \centering
        \def\svgwidth{0.3\textwidth}
\begingroup%
  \makeatletter%
  \providecommand\color[2][]{%
    \errmessage{(Inkscape) Color is used for the text in Inkscape, but the package 'color.sty' is not loaded}%
    \renewcommand\color[2][]{}%
  }%
  \providecommand\transparent[1]{%
    \errmessage{(Inkscape) Transparency is used (non-zero) for the text in Inkscape, but the package 'transparent.sty' is not loaded}%
    \renewcommand\transparent[1]{}%
  }%
  \providecommand\rotatebox[2]{#2}%
  \newcommand*\fsize{\dimexpr\f@size pt\relax}%
  \newcommand*\lineheight[1]{\fontsize{\fsize}{#1\fsize}\selectfont}%
  \ifx\svgwidth\undefined%
    \setlength{\unitlength}{745.51181102bp}%
    \ifx\svgscale\undefined%
      \relax%
    \else%
      \setlength{\unitlength}{\unitlength * \real{\svgscale}}%
    \fi%
  \else%
    \setlength{\unitlength}{\svgwidth}%
  \fi%
  \global\let\svgwidth\undefined%
  \global\let\svgscale\undefined%
  \makeatother%
  \begin{picture}(1,0.12547529)%
    \lineheight{1}%
    \setlength\tabcolsep{0pt}%
    \put(0,0){\includegraphics[width=\unitlength,page=1]{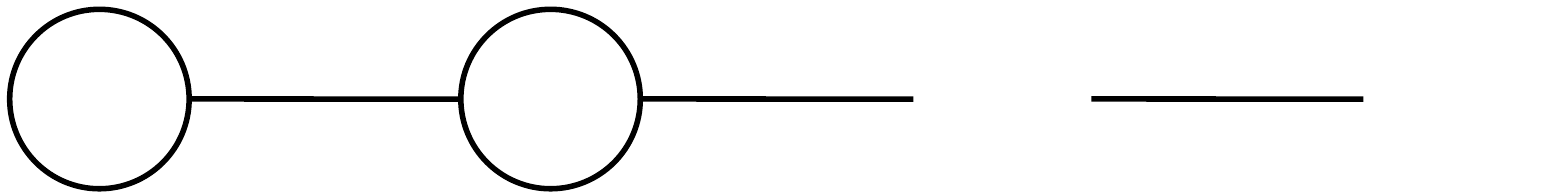}}%
    \put(0.04616404,0.04214006){\color[rgb]{0,0,0}\makebox(0,0)[lt]{\lineheight{1.25}\smash{\begin{tabular}[t]{l}1\end{tabular}}}}%
    \put(0.33706138,0.04179947){\color[rgb]{0,0,0}\makebox(0,0)[lt]{\lineheight{1.25}\smash{\begin{tabular}[t]{l}2\end{tabular}}}}%
    \put(0,0){\includegraphics[width=\unitlength,page=2]{node_graph_path_small.pdf}}%
    \put(0.62861855,0.04220555){\color[rgb]{0,0,0}\makebox(0,0)[lt]{\lineheight{1.25}\smash{\begin{tabular}[t]{l}3\end{tabular}}}}%
    \put(0.91707815,0.04220553){\color[rgb]{0,0,0}\makebox(0,0)[lt]{\lineheight{1.25}\smash{\begin{tabular}[t]{l}4\end{tabular}}}}%
    \put(0,0){\includegraphics[width=\unitlength,page=3]{node_graph_path_small.pdf}}%
  \end{picture}%
\endgroup%

        \caption{Example for a path communication graph and k=3.}
        \label{fig: 2 fig 3-hop neihbours communication graph = path graph}
         \vspace{-0.4cm}
    \end{figure}    
    \noindent Denote the global state with $\vect{x}=[x_1,x_2,x_3,x_4]^\top$. Then, the vectors $\estimate{\vect{x}}{1}{}$, $\estimate{\vect{x}}{}{1}$ and the matrices $\hat{\vect{P}}_1$, $\vect{P}_1$ are defined as:
    \begin{equation}
        \begin{split}
             \estimate{\vect{x}}{1}{} &= \begin{bmatrix}
                                            \estimate{x}{1}{3} \\ \estimate{x}{1}{4}
                                        \end{bmatrix} , \ \                             
            \estimate{\vect{x}}{}{1} = \begin{bmatrix}
                                            \estimate{x}{3}{1} \\ \estimate{x}{4}{1}
                                        \end{bmatrix}, \\ 
           \hat{\vect{P}}_1 &= \begin{bmatrix}
                                            0 && 0 && 1 && 0 \\
                                            0 && 0 && 0 && 1
                                        \end{bmatrix}, \\ 
            \vect{P}_1 &= \begin{bmatrix}
                            0 && 1 && 0 && 0 \\
                        \end{bmatrix}.
        \end{split}
   \end{equation}
    Finally, the set of 3-hop neighbors of agent $1$ is $\Nhop{1}{3} = \{n_1^1,n_1^2 \}=\{3,4\}$.
\end{example}

Suppose the following assumptions hold for the communication graph:
\begin{assumption}
    Each $i\in \mathcal{V}$ knows which are the agents belonging to $\Neigh{i}{}$ and $\Nhop{i}{k}$. 
\end{assumption}
\begin{assumption}
    \label{Assumption on 2-hop communication}
    Each agent $i\in \mathcal{V}$ has access and can propagate at each time instant the state and input of the agents $j \in \Neigh{i}{}$ to its $1$-hop neighbors $\Neigh{i}{}$.
\end{assumption}

Note that, while the last assumption on the state is reasonable if we consider sensing capabilities of the agents, the one on the input requires the communication to be fast enough with negligible time delays.

With the notation presented, the first problem can be formalized as:
\begin{problem}
    \label{Problem state estimation}
    Design a finite-time convergent distributed observer that allows each agent $i \in \mathcal{V}$ with dynamics \eqref{eq: agent's dynamic} to track the state of all $j \in \Nhop{i}{k}$, i.e., such that:
    \begin{equation}
        \exists \ T_{x,i}>0 : \norm{\tilde{\vect{x}}^i(t)} = 0 \ \  \forall  t \geq T_{x,i} \text{ and } \forall  i \in \mathcal{V}.
    \end{equation}
\end{problem}

\section{Distributed \texorpdfstring{$k$}\-hop State Observer}\label{state_observer}
In this section a finite-time distributed observer to solve Problem \ref{Problem state estimation} is introduced. For this purpose assume that the state estimate $\estimate{\vect{x}}{i}{}$ is updated following \eqref{Eq: State observer dynamic}:
\begin{equation}
    \label{Eq: State observer dynamic}
    \begin{split}
        \dot{\hat{\vect{x}}}^i &=\vect{f}(\estimate{\vect{x}}{i}{}) + \vect{A}^{i}\estimate{\vect{x}}{i}{}+\vect{\Omega}^i\vect{G}^i\vect{\xi}^i + \vect{\Theta}^i \sign(\vect{G}^i \vect{\xi}^i) + \estimate{\vect{u}}{i}{} \\
        \vect{\xi}^i &= \sum_{j\in \Neigh{i}{}} [ \hat{\vect{P}}_i(- \hat{\vect{P}}_j^\top  \hat{\vect{P}}_j \hat{\vect{P}}_i^\top \estimate{\vect{x}}{i}{} + \hat{\vect{P}}_i^\top  \hat{\vect{P}}_i \hat{\vect{P}}_j^\top \estimate{\vect{x}}{j}{}) + \\
        &\hspace{1cm}+ \hat{\vect{P}}_i(- \vect{P}_j^\top \vect{P}_j \hat{\vect{P}}_i^\top \estimate{\vect{x}}{i}{} + \hat{\vect{P}}_i^\top  \hat{\vect{P}}_i \vect{P}_j^\top \vect{P}_j \vect{x})],
    \end{split}
\end{equation}
where $\vect{f}(\estimate{\vect{x}}{i}{})$, $\vect{A}^{i}$, $\vect{\Omega}^i$, $ \vect{\Theta}^i$ and $\vect{G}^i$  are defined as:
\begin{equation}
    \label{f estimate definition}
     \vect{f}(\estimate{\vect{x}}{i}{}) =  \left[\state{f}{\top}{}(\estimate{x}{i}{n_1^i}), \dots ,\state{f}{\top}{}(\estimate{x}{i}{n^i_{\eta_i}}) \right]^\top, 
\end{equation}

\begin{equation}
    \label{Eq: matrix A definition}
    \vect{A}^{i} = I_{\eta_i}\otimes A , \vect{\Omega}^i = \Omega^i \otimes I_N  ,  \vect{\Theta}^i = \Theta^i \otimes I_N, \vect{G}^i= I_{\eta_i}\otimes G
\end{equation}
with:
 \begin{equation}
 \label{Omega and theta definition}
    \Omega^i =  \text{diag}(\omega_{n_1^i}, \dots , \omega_{n^i_{\eta_i}}) \ ,\ \Theta^i =  \text{diag}(\theta_{n_1^i} \dots , \theta_{n^i_{\eta_i}}), 
\end{equation}
where $\omega_j \in \mathbb{R}_{\geq 0}$, $\theta_j \in \mathbb{R}_{\geq 0}$ are observer parameters to be tuned $\forall j \in \mathcal{V}$ and where $G \in \mathbb{R}^{N \times N}$ is a positive symmetric matrix to be designed.
In \eqref{Eq: State observer dynamic}, each agent $i \in \mathcal{V}$ updates its state estimate $\estimate{\vect{x}}{i}{}$ of all $l \in \Nhop{i}{k}$ based on the real state information $x_l$ coming from the agents $j \in \Neigh{i}{} \cap \Neigh{l}{}$ : $\hat{\vect{P}}_i( -\vect{P}_j^\top \vect{P}_j \hat{\vect{P}}_i^\top \estimate{\vect{x}}{i}{} + \hat{\vect{P}}_i^\top  \hat{\vect{P}}_i \vect{P}_j^\top \vect{P}_j \vect{x})$ and on the estimation $\estimate{x}{j}{l}$ coming from those $j \in \Neigh{i}{} \cap \Nhop{l}{k}$: $\hat{\vect{P}}_i(- \hat{\vect{P}}_j^\top  \hat{\vect{P}}_j \hat{\vect{P}}_i^\top \estimate{\vect{x}}{i}{} + \hat{\vect{P}}_i^\top  \hat{\vect{P}}_i \hat{\vect{P}}_j^\top \estimate{\vect{x}}{j}{})$ . 

\begin{remark}
\label{equivalence between estimation}
    Since it is always possible for any graph $\mathcal{G}$ and any value of $k$ to find a full rank permutation matrix $T$ of proper dimension, such that:
    \begin{equation}
    \label{eq: equivalence between estimations}
    \left[(\error{\vect{x}}{1}{})^\top, \dots,  (\error{\vect{x}}{n}{})^\top \right]^\top=
    T \left[\error{\vect{x}}{\top}{1},\dots,\error{\vect{x}}{\top}{n}\right]^\top,
\end{equation}
the convergence of $\error{\vect{x}}{}{i}(t)$ as per \eqref{Definition: error on input and state estimation of agent i} implies the one of $\error{\vect{x}}{i}{}(t)$ as per \eqref{Definition: error on input and state estimation}.
\end{remark} 
Therefore, Problem \ref{Problem state estimation} can be reformulated as:

\begin{problem}
    \label{Rewritten problem state convergence}
    Design a finite-time convergent distributed observer that allows each agent $i \in \mathcal{V}$ with dynamics \eqref{eq: agent's dynamic} to track the state of all $j \in \Nhop{i}{k}$, i.e., such that:
    \begin{equation}
        \exists \ T_{x,i} > 0: \ \ \norm{\tilde{\vect{x}}_i(t)} = 0, \ \ \forall t\geq T_{x,i} \text{ and } \forall i \in \mathcal{V}.
    \end{equation}
\end{problem}

Denote with $\tilde{x}_{i}^{l}$ and $\error{u}{l}{i}$ the $l$-th component of vectors $\error{\vect{x}}{}{i}$ and $\error{\vect{u}}{}{i}$, i.e., the estimation errors on the $i$-th agent state and input when the estimate is performed by agent $l \in \Nhop{i}{k}$ and define with $\xi_i^l$ the update on the estimation of the $i$-th agent state done by agent $l$ associated to $\vect{\xi}^l$ in \eqref{Eq: State observer dynamic}.

To prove the observer convergence, we investigate the dynamic of $\tilde{x}_{i}^{l}$ resulting from $\dot{\hat{\vect{x}}}^l$, the definition \eqref{Definition: error on input and state estimation of agent i} and the transformation \eqref{eq: equivalence between estimations}:

\begin{equation}
    \label{eq: state estimation error component}
    \begin{split}
    \dot{\tilde{x}}_{i}^{l} &= \Bar{f}(\estimate{x}{l}{i}) + A \tilde{x}_{i}^{l} - \omega_i G  \xi_i^l -  \theta_i \sign(G \xi_i^l) + \error{u}{l}{i}, \\
    \xi_i^l &=\sum\limits_{k\in(\Neigh{l}{} \cap \Nhop{i}{k})} (\statediff{\tilde{x}}{l}{i}{\tilde{x}}{k}) +\sum\limits_{k\in (\Neigh{l}{} \cap \Neigh{i}{})}  \tilde{x}^{l}_i,
\end{split}
\end{equation}
where $\Bar{f}(\estimate{x}{l}{i}) = f(x_i) -  f(\estimate{x}{l}{i})$.

By defining the vector $\vect{\xi}_i := \left[(\xi_{i}^{n_1^i})^\top, \dots ,(\xi_{i}^{n^i_{\eta_i}})^\top \right]^\top$, we can write:
\begin{equation}
    \label{eq: xi expression}
    \vect{\xi}_i = \left((\kc{L}{i} + \kc{H}{i}) \otimes I_N \right) \error{\vect{x}}{}{i} = (\kc{M}{i}\otimes I_N)\error{\vect{x}}{}{i},
\end{equation}
where:
\begin{enumerate}   
    \item The matrix $\kc{L}{i}$ is the Laplacian matrix of the sub-graph $\mathcal{G}_i = ( \Nhop{i}{k}, \mathcal{E}_i)$ induced by the $k$-hop neighbors of agent $i$, with $\mathcal{E}_i = \{(p,q) \in \mathcal{E}: \{p,q\}\in \Nhop{i}{k}\}$ \cite[pp. 24]{eb184279-05f5-3acc-a300-750c6f4a17e8}.

    \item The matrix $\kc{H}{i}$ is a diagonal matrix of the form:
    \begin{equation}
        \kc{H}{i} = \text{diag}\left(|\Neigh{n_1^i}{}\cap \Neigh{i}{}|, \dots, |\Neigh{n_{\eta_i}^i}{}\cap \Neigh{i}{}|\right).
    \end{equation}
    \item The matrix $M_i^{kc}$ is defined as: 
    \begin{equation}
          \label{eq: Mkc defintion}
        \kc{M}{i} = \kc{L}{i} + \kc{H}{i}.
    \end{equation}
\end{enumerate}
By means of the mixed-product property of the Kronecker product in \cite{graham2018kronecker} and the expression in \eqref{eq: xi expression}, the vector form of \eqref{eq: state estimation error component} results into:
\begin{equation}
    \label{State error dynamic}
    \begin{split}
         \dot{\tilde{\vect{x}}}_{i} &= (\vect{f}(\vect{x}_i) - \vect{f}(\estimate{\vect{x}}{}{i})) + (\vect{A}^i - \omega_i(M^{\text{kc}}_i \otimes G))\tilde{\vect{x}}_i + \\ &- \theta_i \text{sign}((M^{\text{kc}}_i \otimes G) \tilde{\vect{x}}_i) + \error{\vect{u}}{}{i},
    \end{split}
\end{equation}
where $\vect{f}(\estimate{\vect{x}}{}{i})$, $\vect{A}^i$, $G$, $\omega_i$ and $\theta_i$ come from the dynamics in \eqref{Eq: State observer dynamic}, and $\vect{f}(\vect{x}_i) = 1_{\eta_i} \otimes f(x_i)$ .

Before starting the analysis on the finite time stability of the error dynamics, a preliminary result is given in Lemma~\ref{lemma: neighbors}.

\begin{lemma}
    \label{lemma: neighbors}
    Consider an undirected graph $\mathcal{G}=(\mathcal{V},\mathcal{E})$. If $\mathcal{G}$ is connected, then for all $j \in \mathcal{V}$ and all $i \in \Nhop{j}{k}$, $|\Neigh{i}{}\cap \Nhop{j}{k}|> 0 \  \text{ or } \ |\Neigh{i}{}\cap \Neigh{j}{}|> 0$.
   Furthermore for all $j \in \mathcal{V}$ and for each connected component in the sub-graph $G_j = \{ \Nhop{j}{k}, \mathcal{E}_j\}$, there exists at least one agent $i \in \Nhop{j}{k} $ for which $|\Neigh{i}{}\cap \Nhop{j}{k}|> 0 \ \text{ and } \ |\Neigh{i}{}\cap \Neigh{j}{}|> 0$.
\end{lemma}

\begin{proof}
    The first part of Lemma \ref{lemma: neighbors} can be proved by demonstrating that there doesn't exist $j \in\mathcal{V}$ and $i\in \Nhop{j}{k} \ \text{such that} \ |\Neigh{i}{}\cap \Nhop{j}{k}|= 0 \ $ and $ |\Neigh{i}{}\cap \Neigh{j}{}|= 0$.
    Assume by contradiction that there exists $j \in \mathcal{V}$ and $i \in \Nhop{j}{k}$ such that $|\Neigh{i}{}\cap \Nhop{j}{k}|= 0$ and $|\Neigh{i}{}\cap \Neigh{j}{}| = 0$. If $|\Neigh{i}{}\cap \Neigh{j}{}| = 0$ and $i \in \Nhop{j}{k}$ , agent $i$ is connected to agent $j$ by means of a path of length $l>2$. This implies that $|\Neigh{i}{}\cap \Nhop{j}{k}|\neq 0$, which however results to be in contradiction with the initial assumption. 
    
    In a similar way $|\Neigh{i}{}\cap \Nhop{j}{k}|= 0$ and $i \in \Nhop{j}{k}$ implies that agent $i$ is connected to $j$ by means of a path of length $2$. Therefore there must exist an agent $k \in \Neigh{i}{}\cap \Neigh{j}{}$ which contradicts $|\Neigh{i}{}\cap \Neigh{j}{}| = 0$.

    To prove the second statement, suppose by contradiction that for an agent $j \in \mathcal{V}$, all $i$ in a connected component of  $G_j = (\Nhop{j}{k}, \mathcal{E}_j)$ are such that $|\Neigh{i}{}\cap \Neigh{j}{}|= 0$. This would imply that there does not exist a $l$-hop path, with $l\leq k$, between $i$ and $j$, which contradicts the assumption $i \in \Nhop{j}{k}$.
\end{proof}
 
We now present a main result regarding the positive definiteness of the matrix $\kc{M}{i}$ defined in \eqref{eq: Mkc defintion}.
\begin{lemma}
    \label{lemma: Mkc positive definit}
   Consider an undirected graph $\mathcal{G}=(\mathcal{V},\mathcal{E})$. If $\mathcal{G}$ is connected, then:
   \begin{equation}
       \kc{M}{i} \succ 0, \ \forall i \in \mathcal{V} \text{ with } \Nhop{i}{k} \neq \emptyset.
   \end{equation}
\end{lemma}

\begin{proof}
According to \eqref{eq: Mkc defintion}, $\forall i \in \mathcal{V}$, $\kc{M}{i}$ is the sum of two real positive semi-definite symmetric matrices $\kc{L}{i}$ and $\kc{H}{i}$, that according to the definition of Hermitian matrices in \cite[Def. 4.1.1]{Horn_Johnson_2012} result to be Hermitian. Therefore, from \cite[Corollary 4.3.12]{Horn_Johnson_2012}, the eigenvalues of \eqref{eq: Mkc defintion} satisfy $\lambda_j(\kc{L}{i}) \leq \lambda_j(\kc{M}{i})$ for $j = 1, \dots, \eta_i$,
with equality for some $j$ if and only if $B$ is singular and there exists a  nonzero vector $x$ such that $\kc{L}{i}x = \lambda_j(\kc{L}{i})x$, $\kc{H}{i}x = 0$ and $\kc{M}{i} x = \lambda_j(\kc{M}{i})x$.  
Recalling the definition of $\kc{L}{i}$ as the Laplacian matrix of the sub-graph $\mathcal{G}_i = (\Nhop{i}{k}, \mathcal{E}_i)$, we can deduce that its smallest eigenvalue is equal to zero and that its algebraic multiplicity is related to the number of connected components in the graph $\mathcal{G}_i$. For this reason, since the eigenvectors associated to the zero eigenvalue of $\kc{L}{i}$ represent a base for the null space of $\kc{L}{i}$, the positive definiteness of $\kc{M}{i}$ is directly guaranteed if none of the vectors in the null space of $\kc{L}{i}$ is orthogonal to $\kc{H}{i}$. In this case indeed, from \cite[Corollary 4.3.12]{Horn_Johnson_2012}, the following  holds:
\begin{equation}
\label{positive definiteness}
    0 < \lambda_j(\kc{M}{i}) \ \ \   \forall  j = 1, \dots \eta_i.
\end{equation}

Since $\kc{L}{i}$ represents the Laplacian matrix of a graph characterized by several connected components, each eigenvector associated with the zero eigenvalue belongs to the span of the vectors of ones representing the consensus among the agents of the connected sub-graphs. Therefore, since for Lemma \ref{lemma: neighbors} each connected component has at least one associated non-zero element in the diagonal matrix $\kc{H}{i}$, none of the eigenvectors of $\kc{L}{i}$ is orthogonal to $\kc{H}{i}$, leading to the strict inequality in \eqref{positive definiteness} and therefore to the positive definiteness of the matrix $\kc{M}{i}$  defined in \eqref{eq: Mkc defintion}. 
\end{proof}
Thanks to the symmetry and positive definiteness of the matrix $\kc{M}{i}$, the following result can be achieved: 
\begin{lemma}
    \label{lemma on omega and G tuning}
    Consider matrices $\vect{A}^{i}$ and $\kc{M}{i}$ as defined in \eqref{Eq: matrix A definition} and \eqref{eq: Mkc defintion}, respectively. Assume \eqref{Omega parameter tuning}, \eqref{G condition} holds.
    \begin{equation}
        \label{Omega parameter tuning}
        \omega_i \geq \frac{1}{\lambda_{\text{min}}(\kc{M}{i})} \left(1+ \frac{l_f \norm{(\kc{M}{i} \otimes G)}}{\lambda_{\text{min}}(\kc{M}{i}) \lambda_{\text{min}}(G^\top G)}\right)
    \end{equation}
    \begin{equation}
    \label{G condition}
        G^\top A+A^\top G-2G^\top G \prec 0.
    \end{equation}
    Then \eqref{eq: lemma inequality for state observer convergence} holds.
    \begin{equation}
        \label{eq: lemma inequality for state observer convergence}
        (\kc{M}{i} \otimes G)(\vect{A}^i - \omega_i(\kc{M}{i} \otimes G)) + l_f \norm{(\kc{M}{i} \otimes G)}I_{N\eta_i} \prec 0
    \end{equation}

\end{lemma}
\begin{proof}
   Given the positive definiteness of the symmetric matrix $\kc{M}{i}$, it is always possible to find a matrix $T_i \in \mathbb{R}^{\eta_i \times \eta_i}$ with respect to which $\kc{M}{i}$ can be written as $T_i \kc{\Lambda}{i} T_i^\top =\kc{M}{i}$, where $\kc{\Lambda}{i}= \text{diag}(\lambda_1, \dots, \lambda_{\eta_i})$ with $\lambda_j = \lambda_j(\kc{M}{i})$, for all $j=1,\dots,\eta_i$\cite{Horn_Johnson_2012}. 
   By introducing this decomposition in the term $(\kc{M}{i} \otimes G)(\vect{A}^i - \omega_i(\kc{M}{i} \otimes G))$ of \eqref{eq: lemma inequality for state observer convergence}, it becomes:
    \begin{equation}
             \label{eq: starting equation for proof omega tuning}
            (T_i \kc{\Lambda}{i} T_i^\top \otimes G)\vect{A}^i - \omega_i(T_i \kc{\Lambda}{i} T_i^\top \otimes G)(T_i \kc{\Lambda}{i} T_i^\top \otimes G).
    \end{equation}
    Since for matrices $A$, $B$, $C$ and $D$ of appropriate dimensions, the property $(A\otimes B)(C \otimes D) = (AC)\otimes(BD)$ holds, $\vect{A}^i$, $I_{N \eta_i}$ and $(T_i \kc{\Lambda}{i} T_i^\top \otimes G)$ can be rewritten as:
    \begin{equation}
        \label{rewritten matrices Ai and I}
        \begin{split}
            &\vect{A}^i = (T_i \otimes I_N) (I_{\eta_i} \otimes A)( T_i^\top \otimes I_N), \\
            & I_{N\eta_i} =  (T_i \otimes I_N)( T_i^\top \otimes I_N)
        \end{split}
    \end{equation}
    and:
    \begin{equation}
        \label{rewritten M time G}
        (T_i \kc{\Lambda}{i} T_i^\top \otimes G) = (T_i \otimes I_N) (\kc{\Lambda}{i} \otimes G)( T_i^\top \otimes I_N).
    \end{equation}
    Then, with \eqref{rewritten matrices Ai and I} and \eqref{rewritten M time G}, \eqref{eq: lemma inequality for state observer convergence} can be rewritten after some manipulations as:   
    \begin{equation}
        \label{eq: equation that needs to be negative definit}
        \begin{split}
             (T_i \otimes I_N) [(\Lambda_i \otimes GA) - \omega_i(\Lambda_i^2 \otimes G^\top G) + \\ +l_f \norm{(\kc{M}{i} \otimes G)}I_{N\eta_i}]( T_i^\top \otimes I_N) \prec 0,
        \end{split}
    \end{equation}
    which can be studied by neglecting the outer terms $T_i \otimes I_N$ and $T_i^\top \otimes I_N$. As a consequence, \eqref{eq: equation that needs to be negative definit} results into $[(\Lambda_i \otimes GA) - \omega_i(\Lambda_i^2 \otimes G^\top G) +l_f \norm{(\kc{M}{i} \otimes G)}I_{N\eta_i}]$, which is a block diagonal matrix with elements $\lambda_{j}\left(G^{\top} A - \omega_{i} \lambda_{j} G^{\top} G + \frac{l_f \norm{(\kc{M}{i} \otimes G}}{{\lambda_j}} I_N \right)$.
    Therefore, to prove \eqref{eq: equation that needs to be negative definit}, it suffices to prove:
        \begin{equation}
        \label{block diagonal terms to be negative}
        \lambda_{j}\left(G^{\top} A - \omega_{i} \lambda_{j} G^{\top} G + \frac{l_f \norm{(\kc{M}{i} \otimes G}}{{\lambda_j}} I_N \right)\prec 0,
    \end{equation}
    for all eigenvalues $\lambda_{j} \in \sigma(\kc{M}{i})$.
    Then, if the following holds true:
    \begin{equation}
        \omega_i > \frac{1}{\lambda_{min}(\kc{M}{i})} \left(1+\frac{l_f \norm{\kc{M}{i} \otimes G}}{\lambda_{min}(\kc{M}{i})\lambda_{min}(G^\top G)}\right),
    \end{equation}
    the positive term $\frac{l_f \norm{(\kc{M}{i} \otimes G}}{{\lambda_j}} I_N $ in \eqref{eq: equation that needs to be negative definit} is dominated. As a result, by recalling from Lemma \ref{lemma: Mkc positive definit} that $\lambda_{j} \geq 0 \  \forall j= 1, \dots, \eta_i$, if $G^\top A - G^\top G \prec 0 $ , \eqref{block diagonal terms to be negative} is satisfied for all $\lambda_j$, resulting in the validity of \eqref{eq: lemma inequality for state observer convergence}.
    Since the negative definiteness of $G$ follows from the its symmetric part $\frac{1}{2}(G^\top A + A^\top G -2 G^\top G)$, $G$ can be designed as in \eqref{G condition} \cite[pp. 231]{Horn_Johnson_2012}.
\end{proof}
\begin{remark}
    Note that the existence of a matrix $G$ that satisfies \eqref{G condition} is guaranteed by assuming $(A, I_N)$ to be stabilizable and observable \cite[Th. 2]{lewis2014cooperative}.
\end{remark}
By Lemma \ref{lemma on omega and G tuning}, the convergence of the state estimation can be stated in Theorem \ref{Theorem: state observer convergence theorem}.
\begin{theorem}
    \label{Theorem: state observer convergence theorem}
    Consider the multi-agent system \eqref{eq: agent's dynamic}. Suppose that the communication is described by a connected graph $\mathcal{G}$ and that each agent runs the distributed state observer \eqref{Eq: State observer dynamic}.
    For $i \in \mathcal{V}$, consider the error dynamics in \eqref{State error dynamic} and assume that $\norm{K[\error{\vect{u}}{}{i}]}$ is bounded by $d_{\error{\vect{u}}{}{i}}$. Then, $\error{\vect{x}}{}{i}(t)$ as per \eqref{Definition: error on input and state estimation of agent i} reaches the origin in finite time $T_{x,i}>0$ given that the gain $\theta_i$ as per \eqref{Omega and theta definition} is tuned such that:
    \begin{equation}
    \label{tuning of theta parameter}
        \theta_i > \frac{\lambda_{\text{max}}(\kc{M}{i}) \lambda_{\text{max}}(G)}{\lambda_{\text{min}}(\kc{M}{i}) \lambda_{\text{min}}(G)} d_{\error{\vect{u}}{}{i}},
    \end{equation}
    and that $\omega_i$ and $G$ are designed so that conditions \eqref{Omega parameter tuning} and \eqref{G condition} in Lemma \ref{lemma on omega and G tuning} hold.
    Furthermore:
    \begin{equation}
    \label{State convergence time}
        T_{x,i} \leq \frac{\lambda_{\text{max}}(\kc{M}{i}) \lambda_{\text{max}}(G)}{\phi_i} \norm{\error{\vect{x}}{}{i}(0)}
    \end{equation}
    with:
    \begin{equation}
        \label{Phi value}
        \phi_i  = \theta_i \lambda_{\text{min}}(\kc{M}{i}) \lambda_{\text{min}}(G) -\norm{(\kc{M}{i}\otimes G)} \norm{K[\error{\vect{u}}{}{i}]}.
    \end{equation}
\end{theorem}
\begin{proof}
Since the proposed observer and the error dynamics in \eqref{State error dynamic} are discontinuous, non-smooth analysis must be used to prove the finite-time convergence of \eqref{State error dynamic} \cite{317122}, \cite{4048815}. 

Consider a candidate Lyapunov function as the following continuous differentiable function:
\begin{equation}
    \label{Candidate Lyapunov function for state estimation error}
     V_i(\error{\vect{x}}{}{i}) = \frac{1}{2} \error{\vect{x}}{\top}{i} (\kc{M}{i}\otimes G) \error{ \vect{x}}{}{i}.
\end{equation}
Given the continuous differentiability of \eqref{Candidate Lyapunov function for state estimation error}, its time derivative satisfies $\dot{V}_i(\error{\vect{x}}{}{i}) \overset{a.e.}{\in}  \mathring{V}_i(\error{\vect{x}}{}{i})$
where the generalized derivative $\mathring{V}_i(\error{\vect{x}}{}{i})$ assumes the expression:
\begin{equation}
    \label{eq: Lyapunov function derivative}
    \mathring{V}_i(\error{\vect{x}}{}{i}) = \nabla V_i(\error{\vect{x}}{}{i})^\top K[\dot{\tilde{\vect{x}}}_{i}](\error{\vect{x}}{}{i}, \error{\vect{u}}{}{i}),
\end{equation}
where $\nabla V_i(\error{\vect{x}}{}{i})$ denotes the gradient of $V_i(\error{\vect{x}}{}{i})$.
By introducing \eqref{State error dynamic} and the gradient expression, after some manipulations resulting from properties of the Kronecker product and of the set-valued map $K[](\cdot)$ \cite[Th. 1]{4048815}, \eqref{eq: Lyapunov function derivative} can be rewritten as:
\begin{equation}
\label{genralized derivative of the lyapunov function subset}
 \begin{split}
     \mathring{V}_i(\error{\vect{x}}{}{i}) &\subset \error{\vect{x}}{\top}{i} (\kc{M}{i}\otimes G) \left(  \vect{f}(\vect{x}_i) - \vect{f}(\estimate{\vect{x}}{}{i})\right) + \\ &\error{\vect{x}}{\top}{i} (\kc{M}{i}\otimes G)(\vect{A}^i - \omega_i(M^{\text{kc}}_i \otimes G))\tilde{\vect{x}}_i + \\
& - \theta_i \norm{(\kc{M}{i}\otimes G) \error{\vect{x}}{}{i}}_1 + \error{\vect{x}}{\top}{i} (\kc{M}{i}\otimes G) K[\error{\vect{u}}{}{i}].
\end{split}
\end{equation}
Then, by noticing:
\begin{equation}
    \norm{(\kc{M}{i}\otimes G)\error{\vect{x}}{}{i}}_1 \!\geq \!\norm{(\kc{M}{i}\otimes G)\error{\vect{x}}{}{i}} \!\geq \!\lambda_{\text{min}}(\kc{M}{i}) \lambda_{\text{min}}(G) \norm{\error{\vect{x}}{}{i}}
\end{equation}
and that Lemma \ref{lemma on omega and G tuning} holds due to the validity of \eqref{G condition} and \eqref{eq: lemma inequality for state observer convergence}:
\begin{equation}
\begin{split}
    \error{\vect{x}}{\top}{i}\left[l_f\norm{(\kc{M}{i} \otimes G)} I_{N\eta_i} \right . &+ \left .(\kc{M}{i}\otimes G)(\vect{A}^i - \right . \\ &  \left . \omega_i(M^{\text{kc}}_i \otimes G))\right] \tilde{\vect{x}}_i \leq 0,
\end{split}
\end{equation}
the Lyapunov derivative defined in \eqref{Candidate Lyapunov function for state estimation error} can be upper bounded by:
  \begin{equation}
    \label{bound on generalized derivative}
    \mathring{V}_i(\error{\vect{x}}{}{i})\leq - \phi_i \norm{\error{\vect{x}}{}{i}},
\end{equation}
where $\phi_i$ is defined as per \eqref{Phi value}. If $\theta_i$ is designed according to \eqref{tuning of theta parameter},  $\phi_i$ results to be strictly positive, thus proving the convergence of \eqref{State error dynamic}. Furthermore, by recalling the definition of the candidate Lyapunov function in \eqref{Candidate Lyapunov function for state estimation error}, we have $ V_i(\error{\vect{x}}{}{i})^{\frac{1}{2}} \leq \sqrt{\frac{1}{2}\lambda_{\text{max}}(\kc{M}{i}) \lambda_{\text{max}}(G) }\norm{ \error{\vect{x}}{}{i}}$, from which \eqref{bound on generalized derivative} results into: 
\begin{equation}
    \label{expression for finite time}
     \mathring{V}_i(\error{\vect{x}}{}{i}) \leq - V_i(\error{\vect{x}}{}{i})^{\frac{1}{2}} \frac{\phi_i\sqrt{2}}{\sqrt{\lambda_{\text{max}}(\kc{M}{i}) \lambda_{\text{max}}(G)}}.
\end{equation}

By solving \eqref{expression for finite time} with respect to time, we can compute the upper bound on the convergence time $T_{x,i} \leq \frac{\lambda_{\text{max}}(\kc{M}{i}) \lambda_{\text{max}}(G)}{\phi_i} \norm{\error{\vect{x}}{}{i}(0)}$, which guarantees the finite time convergence of \eqref{State error dynamic} \cite{doi:10.1137/S0363012997321358}.
\end{proof} 

Given the equivalence between Problem \ref{Problem state estimation} and Problem \ref{Rewritten problem state convergence}, the following can be stated from Theorem \ref{Theorem: state observer convergence theorem}:
\begin{corollary}
    \label{Corollary on the state estimation conmvergence}
    Consider the multi-agent system \eqref{eq: agent's dynamic}. Suppose that the communication is described by a connected graph $\mathcal{G}$ and that each agent runs the distributed state observer \eqref{Eq: State observer dynamic} under Assumption \ref{Assumption: input assumption}.
    Then, for all $i \in \mathcal{V}$, $\norm{\error{\vect{x}}{i}{}(t)}\leq \norm{\error{\vect{x}}{i}{}(0)} \ \forall t \geq 0$
    and there exists a $T_x > 0$ such that $\norm{\error{\vect{x}}{i}{}(t)} = 0 \ \forall t>T_x$
    with $T_x = \max_{i \in \mathcal{V}}\{T_{x,i}\}$.
\end{corollary}
\begin{proof}
    Given Theorem \ref{Theorem: state observer convergence theorem} and the equivalence  between the convergence of $\error{\vect{x}}{}{i}(t)$ and $\error{\vect{x}}{i}{}(t)$ from Remark \ref{equivalence between estimation}, there exists a time $T_{x,i}$, for all $ i \in \mathcal{V}$ that satisfies $ \norm{\error{\vect{x}}{}{i}(t)} = 0, \ \forall t>T_{x,i}$.
    As a consequence at time $t>T_x$ with $T_x = \max_{i \in \mathcal{V}}\{T_{x,i}\}$, $\norm{\error{\vect{x}}{}{i}(t)} = 0 \ \forall  i\in \mathcal{V}$.
\end{proof}

\section{Distributed \texorpdfstring{$k$}-hop Input Observer}\label{input_observer}
In this section, we present a finite-time distributed input observer to allow each agent $i\in \mathcal{V}$ to estimate the inputs of all $j \in \Nhop{i}{k}$, i.e. $\estimate{\vect{u}}{i}{}(t)$ as per \eqref{definition of input and state estimations done by agent i}. 
For this purpose, consider the following dynamics for the input estimations:
\begin{equation}
    \label{Eq: Input observer dynamic}
    \begin{split}
        \dot{\hat{\vect{u}}}^i &= \vect{\Pi}^i \sign(\vect{\rho}^i)\\
        \vect{\rho}^i &= \sum_{j\in \Neigh{i}{}} [ \hat{\vect{P}}_i(- \hat{\vect{P}}_j^\top  \hat{\vect{P}}_j \hat{\vect{P}}_i^\top \estimate{\vect{u}}{i}{} + \hat{\vect{P}}_i^\top  \hat{\vect{P}}_i \hat{\vect{P}}_j^\top \estimate{\vect{u}}{j}{}) + \\
        &\hspace{1cm}+ \hat{\vect{P}}_i(- \vect{P}_j^\top \vect{P}_j \hat{\vect{P}}_i^\top \estimate{\vect{u}}{i}{} + \hat{\vect{P}}_i^\top  \hat{\vect{P}}_i \vect{P}_j^\top \vect{P}_j \vect{u})],
    \end{split}
\end{equation}
where $\vect{\Pi}^i = \Pi^i \otimes I_N$,  $\Pi^i \in \mathbb{R}^{\eta_i \times \eta_i}$ is a diagonal matrix of the form $\Pi^i = \text{diag}(\pi_{n_1^i}, \dots \pi_{n^i_{\eta_i}})$ and $\pi_{n_j^i} \in \mathbb R_{\ge 0}$ with $j \in \{1, \dots, \eta_i \}$ is a design parameter to be tuned. Similarly to the state estimation case, the convergence of the input estimation error $\tilde{\vect{u}}^i$ can be equivalently formulated in terms of $\tilde{\vect{u}}_i$.
For this purpose, with \eqref{Eq: Input observer dynamic} and following similar manipulations to those performed for the state observer from \eqref{eq: state estimation error component} to \eqref{State error dynamic}, we can show that the estimation errors on the $i$-th agent input behave according to the dynamics:
\begin{equation}
\label{Eq: Input error dynamic}
    \dot{\tilde{\vect{u}}}_i = \dot{\vect{u}}_i - \pi_i \sign((\kc{M}{i}\otimes I_N) \error{\vect{u}}{}{i}),
\end{equation}
with $\kc{M}{i}$ defined as in \eqref{eq: Mkc defintion}.
The convergence behavior of \eqref{Eq: Input error dynamic} can then be formulated as in Theorem \ref{Theorem on input observer}.
\begin{theorem}
    \label{Theorem on input observer}
    Consider the multi-agent system \eqref{eq: agent's dynamic}. Suppose that the communication is described by a connected graph $\mathcal{G}$ and that each agent runs the distributed input observer \eqref{Eq: Input observer dynamic}.
    For $i \in \mathcal{V}$ consider the error dynamics in \eqref{Eq: Input error dynamic} and assume that  $\norm{K[\dot{u}_i]}$ is bounded by $d_{\dot{u}_i}$ as per Assumption \ref{Assumption: input assumption}. Then, $\error{\vect{u}}{}{i}$ reaches $0$ in finite time $T_{u,i} >0 $ given that the gain $\pi_i$ is tuned such that $\pi_i > \frac{\lambda_{\text{max}}(\kc{M}{i})}{\lambda_{\text{min}}(\kc{M}{i}) } \sqrt{\eta_i} d_{\dot{u}_i}$.
    Furthermore:
    \begin{equation}
        \label{input observer convergence time}
        T_{u,i} \leq \frac{\lambda_{\text{max}}(\kc{M}{i})}{\psi_i} \norm{\error{\vect{u}}{}{i}(0)},
    \end{equation}
    with $\psi_i = \left[ \pi_i \lambda_{\text{min}}(\kc{M}{i}) - \norm{(\kc{M}{i}\otimes I_N)} \sqrt{\eta_i} d_{\dot{u}_i} \right]$.
\end{theorem}
\begin{proof}
    The proof follows similar reasoning as the one of Theorem \ref{Theorem: state observer convergence theorem} with $ V_i(\error{\vect{u}}{}{i}) = \frac{1}{2} \error{\vect{u}}{\top}{i} (\kc{M}{i}\otimes I_N) \error{ \vect{u}}{}{i}$ and is not reported here due to space limitation.
\end{proof}
Thanks to the relation between $\tilde{\vect{u}}^i$ and $\tilde{\vect{u}}_i$, which results from \eqref{Definition: error on input and state estimation} and \eqref{Definition: error on input and state estimation of agent i}, and from similar reasoning done for the state estimation error in Remark \ref{equivalence between estimation}, $\tilde{\vect{u}}^i$ satisfies Corollary \ref{Corollary on input estimation error}.
\begin{corollary}
    \label{Corollary on input estimation error}
    Consider the multi-agent system \eqref{eq: agent's dynamic}. Suppose that the communication is described by a connected graph $\mathcal{G}$ and that each agent runs the distributed input observer \eqref{Eq: Input observer dynamic} under Assumption \ref{Assumption on input}. Then, for all $i \in \mathcal{V}$, $\norm{\error{\vect{u}}{i}{}(t)}\leq \norm{\error{\vect{u}}{i}{}(0)}, \ \forall t \geq 0$
    and there exists a $T_u > 0$ such that $\norm{\error{\vect{u}}{i}{}(t)} = 0 \ \forall t>T_u$,
    with $T_u = \max_{i \in \mathcal{V}}\{T_{u,i}\}$.
\end{corollary}

\begin{remark}
    Note that similar as the matrix $M_i$ in \cite{8754717}, $\kc{M}{i}$ results to be positive definite. However, thanks to the smaller number of required estimations, the spectrum of $\kc{M}{i}$ results to be improved in term of estimation requirements. Indeed given the smaller maximum and the higher minimum eigenvalues, the convergence time and observer parameters result to be smaller compared to the results in \cite{8754717}.
    
    Consider for example the graph in Fig. \ref{fig: 2 fig 3-hop neihbours communication graph = path graph} with $k=3$. Given that the eigenvalue of a scalar is unique and is the scalar itself, and given the matrices definition in \cite[(45)]{8754717} and \eqref{eq: Mkc defintion}, $M_2$, $\kc{M}{2}$ and their eigenvalues result into: 
    \begin{equation}
        \begin{split}
            M_2 = \begin{bmatrix}
                1 && -1  && 0 && 0 \\
                -1 && 3 && 0 && -1 \\
                -1 && 0 && 2 && -1 \\
                0 && 0 && -1 && 1
            \end{bmatrix},& \ 
            \kc{M}{2} = 1, \\
            \lambda_{\text{min}}(M_2)= 0.17, \ \lambda_{\text{max}}(M_2)& = 3.96, \\
            \lambda_{\text{min}}(\kc{M}{2})=\lambda_{\text{min}}(\kc{M}{2})&=1.
        \end{split}
    \end{equation}
    Therefore, thanks to the smaller maximum and to the bigger minimum eigenvalues, for fixed $\theta_2$, $\pi_2$, $G$ and $\norm{K[\error{\vect{u}}{2}{}]}$, \eqref{State convergence time} and \eqref{input observer convergence time} demonstrate smaller time convergence upper bounds compared to those obtained in \cite{8754717}.
    \end{remark}

\section{\texorpdfstring{$k$}-hop Estimation-Based Feedback Controller}\label{Closed_loop}
The design of a $k$-hop distributed observer presented in the previous sections allows the control of each agent $i\in\mathcal{V}$ by indirect exploitation of the states of those agents $j\in\Nhop{i}{k}$. Although Corollaries \ref{Corollary on the state estimation conmvergence} and \ref{Corollary on input estimation error} prove the convergence of the two observers, additional analysis needs to be performed for the composite behavior. For this purpose, in Lemma \ref{Lemma:input and state observer combination} we  present the behavior of the $k$-hop estimation-based closed-loop controller. 

\begin{lemma}
    \label{Lemma:input and state observer combination}
    Consider the multi-agent system \eqref{eq: agent's dynamic}.   Suppose that the communication is described by a connected graph $\mathcal{G}$ and that each agent runs the distributed state and input observers \eqref{Eq: State observer dynamic} and \eqref{Eq: Input observer dynamic}. Under the assumption that $\norm{K[\dot{u}_i]}\leq d_{\dot{u}_i} \forall i \in \mathcal{V}$, there exists a $T_u>0$ and $\mathcal{X}>0 $  such that:
    \begin{equation}
        \norm{\error{\vect{x}}{i}{}(t)} < \mathcal{X}, \ \ \ \ \forall t > T_u
    \end{equation}
    with $T_u = \max_{i \in \mathcal{V}}\{T_{u,i}\}$ and $\mathcal{X}$ defined as:
    \begin{equation}
        \mathcal{X} = \max_{i \in  \mathcal{V}} \left\{\sup_{0 \leq \tau \leq T_u} \norm{\error{\vect{x}}{i}{}(t)} \right\}.
    \end{equation}
    Furthermore, there exists $T_{xu} > 0$ such that:
    \begin{equation}
        \norm{\error{\vect{x}}{i}{}(t)} = 0, \ \ \forall t > T_{xu}
    \end{equation}
    with $T_{xu} = T_u + T_x$.
\end{lemma}
\begin{proof}
    Given the validity of Corollary \ref{Corollary on input estimation error}, there exists a time $T_u$ such that $\norm{\error{\vect{u}}{}{i}(t)} = 0$ for all $t>T_u$ and for all $i \in \mathcal{V}$. This implies that, starting from $T_u$, \eqref{tuning of theta parameter} is satisfied independently from $\theta_i$ and that $\error{\vect{x}}{i}{}$ decreases for all $i\in \mathcal{V}$ according to Theorem \ref{Theorem: state observer convergence theorem}, i.e., $\error{\vect{x}}{}{i}(t) <  \error{\vect{x}}{}{i}(T_u)$ for all $t> T_u$.
    
    In order to prove Lemma \ref{Lemma:input and state observer combination} however, further studies are required for the time interval $[0, T_u]$ where there is no guarantee on the validity of \eqref{tuning of theta parameter}. To this end, consider the inequality in \eqref{bound on generalized derivative} and note that if \eqref{tuning of theta parameter} is not valid, $\phi_i$ defined as in \eqref{Phi value} results to be negative. As a result, the Lyapunov function in \eqref{Candidate Lyapunov function for state estimation error} can increase, and so does the state estimation error $\error{\vect{x}}{}{i}$. 
    However, even in this case, since the generalized Lyapunov derivative is upper-bounded from above by a continuous positive function, i.e. $\mathring{V}_i(\error{\vect{x}}{}{i})\leq - \phi_i \norm{\error{\vect{x}}{}{i}}$, and the time interval is finite, the Lyapunov function and therefore $\error{\vect{x}}{}{i}$ would remain finite over $[0, T_u]$. From this reasoning, given the equivalence between $\error{\vect{x}}{i}{}$ and $\error{\vect{x}}{}{i}$, it follows that an upper bound for the state estimation error $\error{\vect{x}}{}{i}$ of any agent $i$ can be found as $\mathcal{X} = \max_{i \in \mathcal{V}}\left\{\ \sup_{0 \leq\tau\leq T_u} \norm{ \error{\vect{x}}{}{i}}\right\}$,
    which is the largest value in norm that an error may have achieved over $[0, T_u]$.
    To prove the last part, note that for any time $t> T_u$, as said at the beginning of the proof, the conditions of Corollary \ref{Corollary on the state estimation conmvergence} holds. Therefore for all $i \in \mathcal{V}$, $\norm{\error{x}{}{i}(t)} < \norm{\error{x}{}{i}(T_u)}\leq \mathcal{X}$ for all $t> T_u$ and $\norm{\error{x}{}{i}(t)} = 0$ for all $t > T_u+T_x$.
\end{proof}
\begin{remark}
    \label{remark on convergence of state independently upon input}
    According to Theorem \ref{Theorem: state observer convergence theorem}, in case of bounded input, the convergence of the state observer is ensured by the possibility of finding an upper bound on $K[\error{\vect{u}}{}{i}]$ from an upper bound on $u_i$. As a result, $\forall i \in \mathcal{V}$ it is possible to tune $\theta_i$ according to \eqref{tuning of theta parameter} and the states estimate dynamics result to be independent on the convergence of the inputs estimation errors. If this is not the case and there exists only an upper bound on $\norm{K[\dot{u}_i]}$, the states observer convergence will depend on the one of the inputs that, given its convergence, will drive the input estimation error toward values for which \eqref{tuning of theta parameter} is satisfied even if the input is not bounded.
\end{remark}

Consider the vectorized version of the multi-agent dynamics in \eqref{eq: agent's dynamic}:
\begin{equation}
    \label{eq: stacked multi-agent system dynamic}
     \dot{\vect{x}} = \vect{f}(\vect{x}) + (I_n \otimes A) \vect{x} + \vect{u},
\end{equation}
where $\vect{x}$ is defined as in \eqref{eq:vector of state}, and the input vector $\vect{u}$ is a general non-linear state-feedback function of the form:
\begin{equation}
\label{input equation closed loop}
    \vect{u} = \vect{q}(\vect{x}) = \left[q_1(\Bar{\vect{x}}_1,\vect{x}^1)^\top,  \dots ,q_n(\Bar{\vect{x}}_n,\vect{x}^n)^\top \right]^\top,
\end{equation}
where $\vect{x}^i$ for each $i$ is given in \eqref{stack vector of real values estimated by agent i} and each $\Bar{\vect{x}}_i$ is the vector containing the state information of agent $i$ and of its $1$-hop neighbors.
To satisfy the condition on the upper bound of the input derivative required for the input observers, let's assume there exists a known upper bound on $K[\dot{q}_i](\cdot)$.

Due to the lack of local information regarding the $k$-hop neighbors' state ${\vect{x}}^{i} \ \forall i \in \mathcal{V}$, the previous controller is implemented by adopting their estimates $\estimate{\vect{x}}{i}{}$ :
\begin{equation}
    \label{eq: feedback controller in unction of estimate and nighbour state}
    \vect{u} = \vect{q}(\Bar{\vect{x}},\hat{\vect{x}}) = \left[q_1(\Bar{\vect{x}}_1,\hat{\vect{x}}^1)^\top,  \dots ,q_n(\Bar{\vect{x}}_n,\hat{\vect{x}}^n)^\top \right]^\top.
\end{equation}
However, by noticing that $\Bar{\vect{x}}_i$ is a selection of the components of $\vect{x}$ and that $\estimate{\vect{x}}{i}{} = \vect{x}^i -\error{\vect{x}}{i}{}$, \eqref{eq: feedback controller in unction of estimate and nighbour state} can be rewritten as:
\begin{equation}
    \label{eq: vectorized local input controller}
    \vect{u} = \vect{q}(\vect{x},\vect{x}-\error{\vect{x}}{}{}),
\end{equation}
where $\error{\vect{x}}{}{} = [\error{\vect{x}}{1 \top}{},\dots, \error{\vect{x}}{n \top}{}]^\top$.

By defining $\Phi(\vect{x}, \error{\vect{x}}{}{}) = \vect{f}(\vect{x}) + (I_n \otimes A) \vect{x} + \vect{q}(\vect{x},\vect{x}-\error{\vect{x}}{}{})$,
\eqref{eq: stacked multi-agent system dynamic} becomes $\dot{\vect{x}} = \Phi(\vect{x}, \error{\vect{x}}{}{})$, where $ \error{\vect{x}}{}{}$ is interpreted as an input disturbance for the system with nominal unforced dynamics $\dot{\vect{x}} = \Phi(\vect{x}, 0_{Nn})$.
\begin{definition}[\cite{SONTAG1995351}]
\label{Definition set-ISS definition fomal}
    A system $\dot{x}=f(x,u)$, with $f:\mathbb{R}^n \times \mathbb{R}^m \rightarrow  \mathbb{R}^n$ is set-Input to State Stable (set-ISS) with respect to a set $\mathcal{A}$ if there exists a $\mathcal{KL}$ function $\beta$ and a $\mathcal{K}$ function $\gamma$ such that, for each initial condition and any locally essentially bounded input $u$ satisfying $\sup_{t\geq 0}\norm{u(t)}\leq \infty$, the following holds:
    \begin{equation}
    \label{set-Iss definition}
        \norm{x(t)}_{\mathcal{A}}\leq \beta(\norm{x(0)}_{\mathcal{A}},t) + \gamma \left(\sup_{0 \leq \tau\leq t}\norm{u(\tau)}\right),
    \end{equation}
    where $\norm{x(t)}_{\mathcal{A}} = \text{dist}(x, \mathcal{A})= \inf_{a \in \mathcal{A}}\{\norm{x-a}\}$.
\end{definition}
\begin{assumption}
    \label{Assumption: on local feedback controller}
    Under perfect state knowledge, the non-linear state-feedback $\vect{u}=\vect{q}(\vect{x})$ as in \eqref{input equation closed loop} ensures convergence of the multi-agent system to an equilibrium representing the system objective, e.g. consensus, formation and flocking.
\end{assumption}
Under Assumption \ref{Assumption: on local feedback controller}, we present the overall stability of the multi-agent system with the designed observer applied. 
\begin{theorem}
    \label{theorem on closed loop stability}
     Consider the multi-agent in \eqref{eq: agent's dynamic}. Suppose that the communication is described by a connected graph $\mathcal{G}$ and that each agent runs the distributed state and input observers \eqref{Eq: State observer dynamic} and \eqref{Eq: Input observer dynamic}. Furthermore, assume that each agent runs the local control input \eqref{eq: vectorized local input controller} under the assumption of bounded $K[\dot{q}_i](\cdot), \forall i \in \mathcal{V}$. Then, if $\Phi(\vect{x}, \error{\vect{x}}{}{})$ is set-ISS with respect to a set $\mathcal{A}$ representing the system objective and Assumption \ref{Assumption: on local feedback controller} holds, the multi-agent system reaches an equilibrium representing the team objective.
\end{theorem}
\begin{proof}
    Given the validity of Lemma \ref{Lemma:input and state observer combination}, there exists an upper bound $\mathcal{X}> 0$, such that for all agents we have $\norm{\error{\vect{x}}{i}{}(t)} < \mathcal{X}$. As a result $\norm{\error{\vect{x}}{}{}(t)} < \sqrt{n} \mathcal{X}$, $\forall t$. Furthermore, Lemma \ref{Lemma:input and state observer combination} guarantees the existence of $T_{xu}>0$ such that $\norm{\error{\vect{x}}{i}{}(t)}= 0, \forall t > T_{xu}$. By exploiting the set-ISS assumption on $\Phi(\vect{x}, \error{\vect{x}}{}{})$ and that from $t=T_{xu}$ the multi-agent system evolves from $\vect{x}(T_{xu})$ under the dynamics $\Phi(\vect{x},0_{Nn} )$, we can conclude that $\norm{\vect{x}(t)}_{\mathcal{A}}\leq \beta(\norm{\vect{x}(T_{xu})}_{\mathcal{A}},t-T_{xu}), \forall t > T_{xu}$. As a result, thanks to the convergence of $\beta(\norm{\vect{x}(T_{xu})}_{\mathcal{A}},t-T_{xu})$ to $0$ resulting from the set-ISS definition, an equilibrium is achieved and the convergence toward the objective represented by the set $\mathcal{A}$ is guaranteed. 
\end{proof}

\section{Simulations}\label{Simulations}
\begin{figure}[t!]
    \vspace{0.2cm}
    \centering
    \def\svgwidth{0.2\textwidth}
    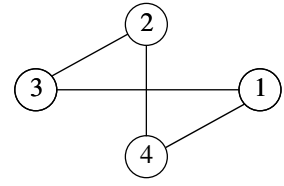
    \caption{Graph $\mathcal{G}_T$ used for consensus.}
    \label{fig:Graph Gt used for consensus}
       \vspace{-0.4cm}
\end{figure}
Consider a multi-agent system composed of $n=4$ agents communicating according to the path graph  $\mathcal{G}_C=(\mathcal{V},\mathcal{E}_C)$ depicted in Fig. \ref{fig: 2 fig 3-hop neihbours communication graph = path graph}.
Assume each agent behaves according to the single integrator dynamic $\dot{x}_i = u_i$,
where $x_i\in [x_{\text{min}}, x_{\text{max}}] \subset \mathbb{R}^2$ and the input $u_i$ is designed in order to drive the agents towards consensus by exploiting only the edges of the graph $\mathcal{G}_T=(\mathcal{V},\mathcal{E}_T)$ shown in Fig. \ref{fig:Graph Gt used for consensus}, i.e.:
\begin{equation}
    \label{Simulation input expression}
    u_i(t) = \sum_{j\in \mathcal{N}^{CT}_i }(x_j(t)-x_i(t)) + \sum_{j\in \mathcal{N}^T_i/\mathcal{N}^{CT}_i}(\hat{x}^i_j(t)-x_i(t)),
\end{equation}
where $\mathcal{N}^C_i$ and $\mathcal{N}^T_i$ are the neighbors of agent $i \in \mathcal{V}$ respectively in graph $\mathcal{G}_C$ and $\mathcal{G}_T$ and $\mathcal{N}^{CT}_i=(\mathcal{N}^C_i  \cap  \mathcal{N}^T_i)$.

It is worth noticing that the problem under study differs from the classical consensus problem, as edges not belonging to the communication graph are exploited to achieve the team objective.
Given the boundedness of the state and of the state estimations, it is possible to prove the existence of an upper bound for the input, i.e., $ u_i(t) \leq \card{N}{T}{i} d_{\text{max}}$
with $d_{\text{max}} = x_{\text{max}}-x_{\text{min}}$.
This, according to the consideration performed in Remark \ref{remark on convergence of state independently upon input}, implies that the state observation of each agent converges independently on the input observer behavior.
To claim the applicability of Theorem \ref{theorem on closed loop stability} to this case study, Assumption \ref{Assumption: on local feedback controller} and the set-ISS property of $u_i$ with respect to the set $\mathcal{A}$ representing the state consensus along the $2$ state components needs to be checked. For this purpose note that, given  the connectivity of $\mathcal{G}_T=(\mathcal{V},\mathcal{E}_T)$,  the input $u_i(t) = \sum_{j\in \mathcal{N}^T_i}(x_j(t)-x_i(t))$
guarantees the convergence of the multi-agent system towards consensus \cite{eb184279-05f5-3acc-a300-750c6f4a17e8}. 
Furthermore, since the input $u_i$ can be rewritten as $u_i(t)=\sum_{j\in \mathcal{N}^{T}_i }(x_j(t)-x_i(t)) - v_i$,
with $v_i(t)=\sum_{j\in \mathcal{N}^T_i/\mathcal{N}^{CT}_i} \error{x}{i}{j}$
bounded, it is possible to prove that the vectorized state dynamics $\dot{\vect{x}}(t)= - L_{T} \vect{x}(t)- \vect{v}(t)$,
where $L_T$ is the Laplacian matrix of the graph $\mathcal{G}_T$ and $\vect{v}(t)=[v_1(t), \dots, v_n(t)]^T$, fulfill the following:
\begin{equation}
    \label{iss simulation}
    \norm{\vect{x}(t)}_{\mathcal{A}} \leq e^{-\lambda_2(L_T)t} \norm{\vect{x}(0)}_{\mathcal{A}} + \frac{1}{\lambda_2(L_T)} \sup_{0\leq \tau\leq t}\norm{\vect{v}(\tau)},
\end{equation}
where $\lambda_2(L_T)$ is the minimum eigenvalue greater than $0$.  Then, given the convergence of the state observer, \eqref{iss simulation} is consistent with the set-ISS definition. 
As a result, Theorem \ref{theorem on closed loop stability} holds and the state and input observers with $k=3$ can be adopted to control the system towards consensus. 

For the purpose of the simulations, a sampling time $dt = 10^{-3} s$ and parameters satisfying Theorems \ref{Theorem: state observer convergence theorem} and \ref{Theorem on input observer} have been chosen. 
Fig. \ref{Simulation 1} shows the results obtained with design parameters: $g= 20$, $\omega_1 = \omega_4 = 2.62$, $\omega_2 = \omega_3 = 1.0$, $\theta_1 = \theta_4 = 3.4$, $\theta_2 = \theta_3 = 0.5$, $\pi_1= \pi_4= 9.7$, $\pi_2= \pi_3 =1.0$ as per \eqref{Eq: State observer dynamic} and \eqref{Eq: Input observer dynamic}.
While the agents input vector is initialized by means of \eqref{Simulation input expression} and according to the states and state estimation information, the estimated input vector is initialized to zero for every agent. As introduced in Remark \ref{remark on convergence of state independently upon input}, thanks to the bounded inputs, the states estimations converge allowing the agents to achieve consensus independently from the input observer dynamics.   

\begin{figure}[t!]
    \centering
    \def\svgwidth{1\textwidth}
    \includegraphics[width=1\linewidth]{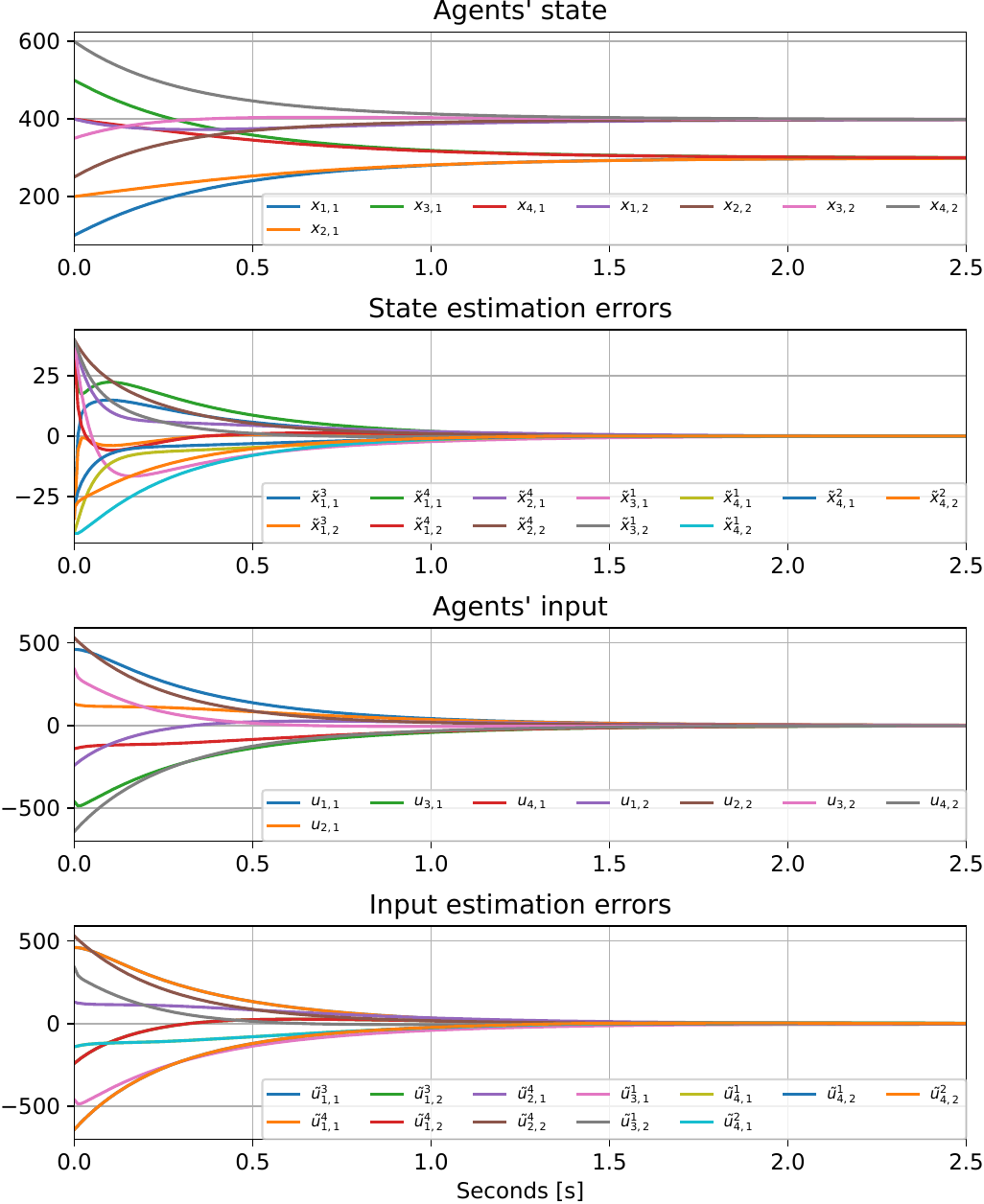}
    \caption{Simulation results with $\pi_1= \pi_4= 9.7$, $\pi_2= \pi_3 =1.0$ as designed parameters for the input observer in \eqref{Eq: Input observer dynamic}.}
    \label{Simulation 1}
    \vspace{-0.6cm}
\end{figure}

\section{Conclusion and Future Work}\label{Conclusion_and_Future_work}
We proposed a communication based $k$-hop distributed observer in which each agent estimates only the states and the inputs of those agents within $k$-hop distance according to the communication graph. The distributed state and input observers result to be finite time convergent and provide state estimations that, under set-ISS condition of the feedback control law, can be used to drive the agents towards an equilibrium representing the team objective. 

As presented in Section \ref{Preliminaries}, while Assumption \ref{Assumption on 2-hop communication} is reasonable for the state if agents are equipped with sensors, it seams more restricting for what concerns the input. For this reason, in addition to study possible disturbance effects, future works will be oriented toward analyzing how the delays on $2$-hop input propagation may affect the observer convergence.


\bibliographystyle{IEEEtran}
\bibliography{references} 
\end{document}